\documentclass[a4paper,USenglish,cleveref, autoref]{lipics-v2019}

\usepackage{color}
\definecolor{keywordcolor}{rgb}{0.7, 0.1, 0.1}   
\definecolor{tacticcolor}{rgb}{0.1, 0.2, 0.6}    
\definecolor{commentcolor}{rgb}{0.4, 0.4, 0.4}   
\definecolor{symbolcolor}{rgb}{0.0, 0.1, 0.6}    
\definecolor{sortcolor}{rgb}{0.1, 0.5, 0.1}      
\definecolor{attributecolor}{rgb}{0.7, 0.1, 0.1} 

\usepackage{mathabx}

\newcommand{\B}{\mathbb{B}}
\newcommand{\lil}{\lstinline}
\newcommand{\N}{\mathbb{N}}

\theoremstyle{theorem}
\newtheorem{thm}[theorem]{Theorem}
\theoremstyle{definition}
\newtheorem{defn}[definition]{Definition}

\title{A formalization of forcing and the unprovability of the continuum hypothesis\footnote{This is an extended preprint of a paper which was submitted to ITP 2019.}} 

\titlerunning{A formalization of forcing and the unprovability of the continuum hypothesis}

\author{Jesse Michael Han\footnote{Corresponding author.}}{Department of Mathematics, University of Pittsburgh \and \url{https://www.pitt.edu/~jmh288}}{jessemichaelhan@gmail.com}{}{}

\author{Floris van Doorn}{Department of Mathematics, University of Pittsburgh \and \url{http://florisvandoorn.com/}}{fpvdoorn@gmail.com}{}{}

\authorrunning{J.\,M. Han and F.\, van Doorn}

\Copyright{Jesse Michael Han and Floris van Doorn}



\category{}

\relatedversion{}

\supplement{\url{https://github.com/flypitch/flypitch}}

\funding{Both authors were supported by the Sloan Foundation, grant G-2018-10067.}

\acknowledgements{We thank the members of the Pitt-CMU Lean group, particularly Simon Hudon, Jeremy Avigad, Mario Carneiro, and Tom Hales for their feedback and suggestions; we are also grateful to Dana Scott and John Bell for their advice and correspondence.}

\nolinenumbers 

\hideLIPIcs  

\EventEditors{John Q. Open and Joan R. Access}
\EventNoEds{2}
\EventLongTitle{42nd Conference on Very Important Topics (CVIT 2016)}
\EventShortTitle{CVIT 2016}
\EventAcronym{CVIT}
\EventYear{2016}
\EventDate{December 24--27, 2016}
\EventLocation{Little Whinging, United Kingdom}
\EventLogo{}
\SeriesVolume{42}
\ArticleNo{23}

\begin{document}

\maketitle

\begin{abstract}
We describe a formalization of forcing using Boolean-valued models in the Lean 3 theorem prover, including the fundamental theorem of forcing and a deep embedding of first-order logic with a Boolean-valued soundness theorem. As an application of our framework, we specialize our construction to the Boolean algebra of regular opens of the Cantor space $2^{\omega_2 \times \omega}$ and formally verify the failure of the continuum hypothesis in the resulting model.
\end{abstract}

\section*{Introduction}
The continuum hypothesis states that there are no sets strictly larger than the countable natural numbers and strictly smaller than the uncountable real numbers. It was introduced by Cantor \cite{cantor1878beitrag} in 1878 and was the very first problem on Hilbert's list of twenty-three outstanding problems in mathematics. G\"odel \cite{godel1938consistency} proved in 1938 that the continuum hypothesis was consistent with $\mathsf{ZFC}$, and later conjectured that the continuum hypothesis is independent of $\mathsf{ZFC}$, i.e. neither provable nor disprovable from the $\mathsf{ZFC}$ axioms. In 1963, Paul Cohen developed \emph{forcing} \cite{cohen-the-independence-of-the-continuum-hypothesis-1,cohen1964independence2}, which allowed him to prove the consistency of the negation of the continuum hypothesis, and therefore complete the independence proof. For this work, which marked the beginning of modern set theory, he was awarded a Fields medal---the only one to ever be awarded for a work in mathematical logic.

In this paper we discuss the formalization of a Boolean-valued model of set theory where the continuum hypothesis fails. 
The work we describe is part of the Flypitch project, which aims to formalize the independence of the continuum hypothesis. Our results mark a major milestone towards that goal.

Our formalization is written in the Lean 3 theorem prover. Lean is an interactive proof assistant under active development at Microsoft Research \cite{de2015lean, sebastian1}. It implements the Calculus of Inductive Constructions and has a similar metatheory to Coq, adding definitional proof irrelevance, quotient types, and a noncomputable choice principle. 
Our formalization makes as much use of the expressiveness of Lean's dependent type theory as possible, using constructions which are impossible or unwieldy to encode in HOL, much less ZF: Lean's ordinals and cardinals, which are defined as equivalence classes of well-ordered types, live one universe level up and play a crucial role in the forcing argument; the models of set theory we construct require as input an entire universe of types; our encoding of first-order logic uses parametrized inductive types to equate type-correctness with well-formedness, eliminating the need for separate well-formedness proofs.

The method of forcing with Boolean-valued models was developed by Solovay and Scott in '65-'66 \cite{scott1967proof,scott-solovay} as a simplification of Cohen's method. Some of these simplifications were incorporated by Shoenfield \cite{shoenfield1971unramified} into a general theory of forcing using partial orders, and it is in this form that forcing is usually practiced. While both approaches have essentially the same mathematical content (see e.g. \cite{kunen2014set, jech2013set, moore2019method}), there are several reasons why we chose Boolean-valued models for our formalization:

\begin{itemize}
\item \textbf{Modularity.} The theory of forcing with Boolean-valued models cleanly splits into several components (a general theory of Boolean-valued semantics for first-order logic, a library for calculations inside complete Boolean algebras, the construction of Boolean-valued models of set theory, and the specifics of the forcing argument itself) which could be formalized in parallel and then recombined.

\item \textbf{Directness.} For the purposes of an independence proof, the Boolean-valued soundness theorem eliminates the need to produce a two-valued model. This approach also bypasses any requirement for the reflection theorem/L\"owenheim-Skolem theorems, Mostowski collapse, countable transitive models, or genericity considerations for filters.

\item \textbf{Novelty and reusability.} As far as we were able to tell, the Boolean-valued approach to forcing has never been formalized. Furthermore, while for the purposes of an independence proof, forcing with Boolean-valued models and forcing with countable transitive models accomplish the same thing, a general library for Boolean-valued semantics of a deeply embedded logic could be used for formal verification applications outside of set theory, e.g. to formalize the Boolean-valued semantics of stochastic $\lambda$-calculus \cite{scott2014stochastic, bacci2018boolean}.

  \item \textbf{Amenability to structural induction.} As with Coq, Lean is able to encode extremely complex objects and reason about their specifications using inductive types. However, the user must be careful to choose the encoding so that properties they wish to reason about are accessible by structural induction, which is the most natural mode of reasoning in the proof assistant. After observing (1) that the Aczel-Werner encoding of $\mathsf{ZFC}$ as an inductive type is essentially a special case of the recursive \emph{name} construction from forcing (c.f. Section \ref{sect:bvm}), and (2) that the automatically-generated induction principle for that inductive type \emph{is} $\in$-induction, it is easy to see that this encoding can be modified to produce a Boolean-valued model of set theory where, again, $\in$-induction comes for free.
\end{itemize}

We briefly outline the rest of the paper. In Section \ref{sect:outline} we outline the method of Boolean-valued models and sketch the forcing argument. Section \ref{sect:fol} discusses a deep embedding of first-order logic, including a proof system and the Boolean-valued soundness theorem. Section \ref{sect:bvm} discusses our construction of Boolean-valued models of set theory. Section \ref{sect:forcing} describes the formalization of the forcing argument and the construction of a suitable Boolean algebra for forcing $\neg\mathsf{CH}$. Section \ref{sect:ccc} describes the formalization of some transfinite combinatorics. We conclude with a reflection on our formalization and an indication of future work.

\section{Outline of the proof}
\label{sect:outline}

$\mathsf{ZFC}$ is a collection of first-order sentences in the language of a single binary relation $\{\in\}$, used to axiomatize set theory. The continuum hypothesis can be written in this fashion as a first-order sentence $\mathsf{CH}$. A proof of $\mathsf{CH}$ is a finite list of deductions starting from $\mathsf{ZFC}$ and ending at $\mathsf{CH}$. 
The soundness theorem says that provability implies satisfiability, i.e. if $\mathsf{ZFC} \vdash \mathsf{CH}$, then $\mathsf{CH}$ interpreted in any model of $\mathsf{ZFC}$ is true. Taking the contrapositive, we can demonstrate the unprovability (equivalently, the consistency of the negation) of $\mathsf{CH}$ by exhibiting a single model where $\mathsf{CH}$ is not true.

A model of a first-order theory $T$ in a language $L$ is in particular a way of assigning $\mathsf{true}$ or $\mathsf{false}$ in a coherent way to sentences in $L$. Modulo provable equivalence, the sentences form a Boolean algebra and ``coherent'' means the assignment is a Boolean algebra homomorphism (so $\lor$ becomes join, $\forall$ becomes infimum, etc.) into $\mathbf{2} = \{\mathsf{true}, \mathsf{false}\}$. The soundness theorem ensures that this homomorphism $v$ sends a proof $\phi \vdash \psi$ to an inequality $v(\phi) \leq v(\psi)$. $\mathbf{2}$ may be replaced by any complete Boolean algebra $\B$, where the top and bottom elements $\top, \bot$ take the place of $\mathsf{true}$ and $\mathsf{false}$. It is straightforward to extend this analogy to a $\B$-valued semantics for first-order logic, and in this generality, the soundness theorem now says that for any such $\B$, if $\mathsf{ZFC} \vdash \mathsf{CH}$, then for any $\B$-valued structure where all the axioms of $\mathsf{ZFC}$ have truth-value $\top$, $\mathsf{CH}$ does also. Then as before, to demonstrate the consistency of the negation of $\mathsf{CH}$ it suffices to find just one $\mathbb{B}$ and a single $\mathbb{B}$-valued model where $\mathsf{CH}$ is not ``true''.

This is where forcing comes in. Given a universe $V$ of set theory containing a Boolean algebra $\B$, one constructs in analogy to the cumulative hierarchy a new $\B$-valued universe $V^\B$ of set theory, where the powerset operation is replaced by taking functions into $\B$. Thus, the structure of $\B$ informs the decisions made by $V^\B$ about what subsets, hence functions, exist among the members of $V^\B$; the real challenge lies in selecting a suitable $\B$ and reasoning about how its structure affects the structure of $V^\B$. While $V^\B$ may vary wildly depending on the choice of $\B$, the original universe $V$ always embeds into $V^\B$ via an operation $x \mapsto \check{x}$, and while the passage of $x$ to $\check{x}$ may not always preserve its original properties, properties which are definable with only bounded quantification are preserved; in particular, $V^\B$ thinks $\check{\mathbb{N}}$ is $\mathbb{N}$.

To force the negation of the continuum hypothesis, we use the Boolean algebra $\B := \operatorname{RO}(2^{\aleph_2 \times \mathbb{N}})$ of regular opens of the Cantor space $2^{\aleph_2 \times \mathbb{N}}$. For each $\nu \in \aleph_2$, we associate the $\B$-valued characteristic function $\chi_\nu : \mathbb{N} \to \B$ by $n \mapsto \{f \operatorname{|} f(\nu, n) = 1\}$. This induces what $V^\B$ thinks is a new subset $\widetilde{\chi_{\nu}} \subseteq \mathbb{N}$, called a \emph{Cohen real}, and furthermore, simultaneously performing this construction on all $\nu \in \aleph_2$ induces what $V^\B$ thinks is a function from $\check{\aleph_2} \to \mathcal{P}(\mathbb{N})$. After showing that $V^\B$ thinks this function is injective, to finish the proof it suffices to show that $x \mapsto \check{x}$ preserves cardinal inequalities, as then we will have squeezed $\check{\aleph_1}$ properly between $\mathbb{N}$ and $\mathcal{P}(\mathbb{N})$. This is really the technical heart of the matter, and relies on a combinatorial property of $\B$ called the \emph{countable chain condition} (CCC), the proof of which requires a detailed combinatorial analysis of the basis of the product topology for $2^{\aleph_2 \times \mathbb{N}}$; we handle this with a general result in transfinite combinatorics called the \emph{$\Delta$-system lemma}.

So far we have mentioned nothing about how this argument, which is wholly set-theoretic, is to be interpreted inside type theory. To do this, it was important to separate the mathematical content from the metamathematical content of the argument. While our objective is only to produce a model of $\mathsf{ZFC}$ satisfying certain properties, traditional presentations of forcing are careful to stay within the foundations of $\mathsf{ZFC}$, emphasizing that all arguments may be performed internal to a model of $\mathsf{ZFC}$, etc., and it is not immediately clear what parts of the argument use that set-theoretic foundation in an essential way 
and require modification in the passage to type theory. Our formalization clarifies some of these questions.

Finally, when working with Boolean-valued models, it is profitable to keep in mind the following analogy, developed by Scott in \cite{scott1967proof}. A ready supply of complete Boolean algebras $\B$ is obtained by taking the measure algebra of a probability space and quotienting by the ideal of events of measure zero. Let $\mathbf{M}$ be a $\B$-valued structure. A unary $\B$-valued predicate $\phi$ on $\mathbf{M}$ assigns an event to every element $m$ of $\mathbf{M}$, whose measure we can think of as being the probability that $\phi(m)$ is true. Specializing to the language of set theory, we can attach to every $m : \mathbf{M}$ an ``indicator function'' $\lambda x, x \in m$ which assigns to every $x$ a probability that it is actually a member of $m$. Thus, by virtue of extensionality, we may think of the elements of a $\B$-valued model of$\mathsf{ZFC}$as being ``set-valued random variables'', or ``random sets''\footnote{In this analogy, given a universe of random sets, the purpose of the generic filter or ultrafilter in forcing is then to simultaneously evaluate the outcomes of the random variables, collapsing them into an ordinary universe of sets.}; see \cite{scott1967proof} and \cite{moore2019method} for details.

\paragraph*{Sources} Our strategy for constructing a Boolean-valued model in which $\mathsf{CH}$ fails is a synthesis of the proofs in the textbooks of Bell (\cite{bell2011set}, Chapter 2) and Manin (\cite{manin2009course}, Chapter 8). For the $\Delta$-system lemma, we follow Kunen (\cite{kunen2014set}, Chapters 1 and 5).

\paragraph*{Viewing the formalization}
The code blocks in this paper were taken directly from our formalization, but for the sake of formatting and readability, we sometimes omit or modify universe levels, type ascriptions, and casts. We refer the interested reader to our repository,\footnote{\url{https://github.com/flypitch/flypitch}} which contains a guide on compiling and navigating the source files of the project. In particular, there is a summary file \lil{summary.lean} containing \lil{#print} statements of important definitions and duplicated proofs of the main theorems.

\section{First-order logic}
\label{sect:fol}
The starting point for first-order logic is a \emph{language} of relation and function symbols. We represent a language as a pair of $\N$-indexed families of types, each of which is to be thought of as the collection of relation (resp. function) symbols stratified by arity:
\begin{lstlisting}
structure Language : Type (u+1) :=
(functions : ℕ → Type u) (relations : ℕ → Type u)
\end{lstlisting}
\subsection{(Pre)terms, (pre)formulas}
The main novelty of our implemenation of first-order logic is the use of \emph{partially applied} terms and formulas, encoded in a parametrized inductive type where the $\N$ parameter measures the difference between the arity and the number of applications. The benefit of this is that it is impossible to produce an ill-formed term or formula, because type-correctness is equivalent to well-formedness. This eliminates the need for separate well-formedness proofs.

Fix a language $L$. We define the type of \textbf{preterms} as follows:
\begin{lstlisting}
inductive preterm : ℕ → Type u
| var {} : ∀ (k : ℕ), preterm 0
| func : ∀ {l : ℕ} (f : L.functions l), preterm l
| app : ∀ {l : ℕ} (t : preterm (l + 1)) (s : preterm 0), preterm l
\end{lstlisting}
We use de Bruijn indices to avoid variable shadowing. A member of \lil{preterm n} is a partially applied term. If applied to \lil{n} terms, it becomes a term. Every element of \lil{preterm L 0} is a well-formed term. We use this encoding to avoid mutual or nested inductive types, since those are not too convenient to work with in Lean.

The type of \textbf{preformulas} is defined similarly:
\begin{lstlisting}
inductive preformula : ℕ → Type u
| falsum {} : preformula 0 --  notation `⊥`
| equal (t₁ t₂ : term L) : preformula 0 -- notation `≃`
| rel {l : ℕ} (R : L.relations l) : preformula l
| apprel {l : ℕ} (f : preformula (l + 1)) (t : term L) : preformula l
| imp (f₁ f₂ : preformula 0) : preformula 0 -- notation ⟹
| all (f : preformula 0) : preformula 0 -- notation `∀'`
-- ¬ f := f ⟹ ⊥, notation `∼f`
-- ∃ f := ∼ ∀' ∼f, notation `∃' f`
\end{lstlisting}

A member of \lil{preformula n} is a partially applied formula. If applied to \lil{n} terms, it becomes a formula. Implication is the only binary connective. Since we use classical logic, we can define the other connectives from implication and falsum. Similarly, universal quantification is our only quantifier.

Our proof system is a natural deduction calculus, and 
all rules are motivated to work well with backwards-reasoning:

\begin{lstlisting}
inductive prf : set (formula L) → formula L → Type u
| axm     {Γ A} (h : A ∈ Γ) : prf Γ A
| impI    {Γ} {A B} (h : prf (insert A Γ) B) : prf Γ (A ⟹ B)
| impE    {Γ} (A) {B} (h₁ : prf Γ (A ⟹ B)) (h₂ : prf Γ A) : prf Γ B
| falsumE {Γ} {A} (h : prf (insert ∼A Γ) ⊥) : prf Γ A
| allI    {Γ A} (h : prf Γ A) : prf Γ (∀' A)
| allE₂   {Γ} A t (h : prf Γ (∀' A)) : prf Γ (A[t // 0])
| ref     (Γ t) : prf Γ (t ≃ t)
| subst₂  {Γ} (s t f) (h₁ : prf Γ (s ≃ t)) (h₂ : prf Γ (f[s // 0])) :
          prf Γ (f[t // 0])
\end{lstlisting}

A member of \lil{prf Γ A} is a proof tree encoding a derivation of $A$ from $\Gamma$. Note that \lil{prf} is \lil{Type}- instead of \lil{Prop}-valued, so different members of \lil{prf Γ A} are not definitionally equal.

\subsection{Completeness}
As part of our formalization of first-order logic, we completed a verification of the G\"odel completeness theorem. Although our present development of forcing did not require it, we anticipate that it will useful later to e.g. prove the downward L\"owenheim-Skolem theorem for extracting countable transitive models. Like soundness, it also serves as a proof-of-concept and stress-test of our chosen encoding of first-order logic.

For our formalization, we chose the Henkin-style approach of constructing a canonical term model. In order to perform the argument, which normally involves modifying the language ``in place'' to iteratively add new constant symbols, we had to adapt it to type theory. Since our languages are represented by pairs of indexed types instead of sets, we cannot really modify them in-place with new constant symbols. Instead, at each step of the construction, we must construct an entirely new language in which the previous one embeds, and in the limit we must compute a directed colimit of types instead of a union. This construction induces similar constructions on terms and formulas, and completing the argument requires reasoning with all of them. As a result of our design decisions, only a few arguments required anything more than straightforward case-analysis and structural induction. The final statement makes no restrictions on the cardinality of the language:
\begin{lstlisting}
  theorem completeness {L : Language} (T : Theory L) (ψ : sentence L) : T ⊢' ψ ↔ T ⊨ ψ
\end{lstlisting}

\subsection{Boolean-valued semantics for first-order logic}

A \textbf{complete Boolean algebra} is a type $\B$ equipped with the structure of a Boolean algebra and additionally operations $\operatorname{Inf}$ and $\operatorname{Sup}$ (which we write as $\bigsqcap$ and $\bigsqcup$) returning the infimum and supremum of an arbitrary collection of members of $\B$. We use $\sqcap, \sqcup, \implies, \top$, and $\bot$ to denote meet, join, material implication, and top/bottom elements. For more details on complete Boolean algebras, we refer the reader to the textbook of Halmos-Givant \cite{givant2008introduction}.

\begin{defn}\label{def-boolean-valued-structure}
  Fix a language $L$ and a complete Boolean algebra $\B$. 
  A \textbf{$\B$-valued structure} is an instance of the following \lil{structure}:
  \begin{lstlisting}
structure bStructure :=
(carrier : Type u)
(fun_map : ∀{n}, L.functions n → vector carrier n → carrier)
(rel_map : ∀{n}, L.relations n → vector carrier n → 𝔹)
(eq : carrier → carrier → 𝔹)
(eq_refl : ∀ x, eq x x = ⊤)
(eq_symm : ∀ x y, eq x y = eq y x)
(eq_trans : ∀{x} y {z}, eq x y ⊓ eq y z ≤ eq x z)
(fun_congr : ∀{n} (f : L.functions n) (x y : vector carrier n),
  ⨅(map2 eq x y) ≤ eq (fun_map f x) (fun_map f y))
(rel_congr : ∀{n} (R : L.relations n) (x y : vector carrier n),
  ⨅(map2 eq x y) ⊓ rel_map R x ≤ rel_map R y)
\end{lstlisting}
Above, ``\lstinline{⨅(map2 eq x y)}'' means ``the infimum of the list whose $i$th entry is \lil{eq} applied to \lil{x[i]} and \lil{y[i]}''.
\end{defn}
Note that Boolean-valued equality is not really an equivalence relation, but ``$\B$ thinks it is''. One complication which then arises in Boolean-valued semantics is keeping track of the congruence lemmas for formulas. However, as part of the soundness theorem shows, once these extensionality proofs are provided for the basic symbols in the language, they extend by structural induction to all formulas.

\subsection{The soundness theorem}

A soundness theorem says that a proof tree may be replayed to produce an actual proof in the object of truth-values. When the object of truth-values is \lil{Prop}, this says that a proof tree compiles to a proof term. When the object of truth-values is a Boolean algebra, this says that the proof tree becomes an internal implication from the interpretation of the context to the interpretation of the conclusion:

\begin{lstlisting}[gobble=2]
  lemma boolean_soundness {Γ : set (formula L)} {A : formula L}
    (H : Γ ⊢ A) : ∀ M, (⨅γ ∈ Γ, M[γ]) ≤ M[A]
\end{lstlisting}

Of course, we also formalized the ordinary soundness theorem. As a result of our design decisions, the proofs of both the ordinary and Boolean-valued soundness theorems were straightforward structural inductions.


\section{Constructing Boolean-valued models of set theory}
\label{sect:bvm}
Throughout this section, we fix a universe level $u$ and a complete Boolean algebra \lstinline{𝔹 : Type u}.

In set theory (see e.g. Jech \cite{jech2013set} or Bell \cite{bell2011set}), Boolean-valued models are obtained by imitating the construction of the von Neumann cumulative hierarchy via a transfinite recursion where iterations of the powerset operation (taking functions into $\mathbf{2} = \{\operatorname{true}, \operatorname{false}\}$) are replaced by iterations of the ``\lstinline{𝔹}-valued powerset operation'' (taking functions into $\B$).

Since this construction by transfinite recursion does not easily translate into type theory, our construction of Boolean-valued models of set theory is instead a variation on a well-known encoding originally due to Aczel \cite{aczel1978type, aczel1986type, aczel1982type}. This encoding was adapted by Werner \cite{werner1997sets} to encode $\mathsf{ZFC}$ into Coq, whose metatheory is close to that of Lean. Werner's construction was implemented in Lean's \texttt{mathlib} by Carneiro, as part of \cite{mario1}. In this approach, one takes a universe of types \texttt{Type u} as the starting point and then imitates the cumulative hierarchy by constructing the inductive type
\begin{lstlisting}
inductive pSet : Type (u+1)
| mk (α : Type u) (A : α → pSet) : pSet
\end{lstlisting}
The Aczel-Werner encoding is closely related to the recursive definition of \emph{names}, which is used in forcing to construct forcing extensions:

\begin{defn}\label{def-p-name}
Let $P$ be a partial order (which one thinks of as a collection of forcing conditions). A \emph{$P$-name} is a collection of pairs $(y, p)$ where $y$ is a $P$-name and $p : P$.
\end{defn}

 If $P$ consists of only one element, then a $P$-name is specified by essentially the same information as a member of the inductive type \lstinline{pSet} above. Conversely, specializing $P$ to an arbitrary complete Boolean algebra $\B$, we generalize the definition of \lstinline{pSet.mk} so that elements are recursively assigned Boolean truth-values:
\begin{lstlisting}
inductive bSet (𝔹 : Type u) [complete_boolean_algebra 𝔹] : Type (u+1)
| mk (α : Type u) (A : α → bSet) (B : α → 𝔹) : bSet
\end{lstlisting}
Thus \lil{bSet 𝔹} is the type of $\B$-names, and will be the underlying type of our Boolean-valued model of set theory. For convenience, if \lstinline{x : bSet 𝔹} and \lstinline{x := ⟨α, A, B⟩}, we put \lstinline{x.type := α, x.func := A, x.bval := B}.

\subsection{Boolean-valued equality and membership}

In \lil{pSet}, equivalence of sets is defined by structural recursion as follows: two sets $x$ and $y$ are equivalent if and only if for every $w \in x$, there exists a $w' \in y$ such that $w$ is equivalent to $w'$, and vice-versa. Analogously, by translating quantifiers and connectives into operations on $\B$, Boolean-valued equality is defined in the same way:
\begin{lstlisting}
def bv_eq : ∀ (x y : bSet 𝔹), 𝔹
| ⟨α, A, B⟩ ⟨α', A', B'⟩ :=
             (⨅a : α, B a ⟹ ⨆a', B' a' ⊓ bv_eq (A a) (A' a')) ⊓
               (⨅a' : α', B' a' ⟹ ⨆a, B a ⊓ bv_eq (A a) (A' a'))
\end{lstlisting}

We abbreviate \lil{bv_eq} with the infix operator \lil{=ᴮ}. With equality in place, it is easy to define membership by translating ``$x$ is a member of $y$ if and only if there exists a $w$ indexed by the type of $y$ such that $x = w$.'' As with equality, we denote $\B$-valued membership by \lil{∈ᴮ}.

\begin{lstlisting}
def mem : bSet 𝔹 → bSet 𝔹 → 𝔹
| a ⟨α' A' B'⟩ := ⨆a', B' a' ⊓ a =ᴮ A' a'
\end{lstlisting}

\subsection{Automation and metaprogramming for reasoning in $\B$} \label{subsect:proof-language}
As Scott stresses in \cite{scott2008algebraic}, ``A main point ... is that the well-known algebraic characterizations of [complete Heyting algebras] and [complete Boolean algebras] exactly mimic the rules of deduction in the respective logics.'' Indeed, that is really why the Boolean-valued soundness theorem is true. One thinks of the \lil{≤} symbol in an inequality of Boolean truth-values as a turnstile in a proof state: the conjunctands on the left as a list of assumptions in context, and the quantity on the right as the goal. For example, given \lil{a b : 𝔹}, the identity $(a \Rightarrow b) \sqcap a \leq b$ could be proven by unfolding the definition of material implication, but it is really just modus ponens; similarly, given an indexed family \lil{a : I → 𝔹}, \lstinline{⨆i, a i ≤ b ↔ ∀ i, a i ≤ b} is just $\exists$-elimination.

Difficulties arise when the statements to be proved become only slightly more complicated. Consider the following example, which should be  ``\lil{by assumption}'':
\begin{lstlisting}[gobble=2]
  ∀ a b c d e f g: 𝔹, (d ⊓ e) ⊓ (f ⊓ g ⊓ ((b ⊓ a) ⊓ c)) ≤ a
\end{lstlisting}
or slightly less trivially, the following example where the goal is attainable by ``just applying a hypothesis to an assumption''
\begin{lstlisting}[gobble=2]
  ∀ a b c d : 𝔹, (a ⟹ b) ⊓ c ⊓ (d ⊓ a) ≤ b
\end{lstlisting}

There are three ways to deal with goals like these, which approximately describe the evolution of our approach. First, one can try using the basic lemmas in \lil{mathlib}, using the simplifier to normalize expressions, and performing clever rewrites with the deduction theorem.\footnote{The deduction theorem in a Boolean algebra says that for all $a, b$ and $c$, $a \sqcap b \leq c \iff a \leq b \Rightarrow c$.} Second, one can take the LCF-style approach and expand the library of lemmas with increasingly sophisticated derived inference rules. Third, one can make the following observation:

\begin{lemma}[Yoneda lemma for posets]\label{poset-yoneda}
  Let $(P, \leq)$ be a partially ordered set. Let $a \hspace{1mm} b : P$. Then $a \leq b$ if and only if $\forall \Gamma : P, \Gamma \leq a \to \Gamma \leq b$.
\end{lemma}
This is a consequence of the Yoneda lemma for partially ordered sets, and its proof is utterly trivial. However, one side of the equivalence is much easier for Lean to reason with. Take the example which should have been ``\lil{by assumption}''. The following proof, in which the user navigates down the binary tree of nested \lil{⊓}s, will work:
\begin{lstlisting}
example {a b c d e f g : 𝔹} : (d ⊓ e) ⊓ (f ⊓ g ⊓((b ⊓ a)⊓ c)) ≤ a :=
by {apply inf_le_right_of_le, apply inf_le_right_of_le,
    apply inf_le_left_of_le, apply inf_le_right_of_le, refl}
\end{lstlisting}

But if we use the right-hand side of \autoref{poset-yoneda} instead, then after some preprocessing, \lstinline{assumption} will literally work:

\begin{lstlisting}
example {a b c d e f g : 𝔹} : (d ⊓ e) ⊓ (f ⊓ g ⊓((b ⊓ a)⊓ c)) ≤ a :=
by {tidy_context, assumption}
-- `tidy_context` applies `poset_yoneda`, introduces a hypothesis `H`,
-- uses `simp` at H to convert ⊓s to ∧s, and automatically splits
/- Goal state before `assumption`:
[...]
H_right_right_left_left : Γ ≤ b,
H_right_right_left_right : Γ ≤ a
⊢ Γ ≤ a -/
\end{lstlisting}

A key feature of Lean is that it is its own metalanguage, allowing for seamless in-line definitions of custom tactics. This feature was an invaluable asset, as it allowed the rapid development of a custom tactic library for simulating natural-deduction style proofs inside $\B$ after applying \autoref{poset-yoneda}. Boolean-valued versions of natural deduction rules like $\lor$/$\land$-elimination, instantiation of existentials, implication introduction, and even basic automation were easy to write. The result is that the user is able to pretend, with absolute rigor, that they are simply writing proofs in first-order logic while calculations in the complete Boolean algebra are being performed under the hood.

One use-case where automation is crucial is context-specialization. For example, suppose that after preprocessing with \lstinline{poset_yoneda}, the goal is \lstinline{Γ ≤ a ⟹ b}, and one would like to ``introduce the implication'', adding \lstinline{Γ ≤ a} to context and reducing the goal to \lstinline{Γ ≤ b}. This is impossible as stated. Rather, the deduction theorem lets us rewrite the goal to \lstinline{Γ ⊓ a ≤ b}, and now we may add \lstinline{Γ ⊓ a ≤ a}. So we may introduce the implication after all, but at the cost of specializing the context \lstinline{Γ} to the smaller context \lstinline{Γ' := Γ ⊓ a}. But now, in order for the user to continue the pretense that they are merely doing first-order logic, this change of variables must be propagated to the rest of the assumptions which may still be of the form \lstinline{Γ ≤ _}---which is extremely tedious to do by hand, but easy to automate.

\subsection{The fundamental theorem of forcing}

The fundamental theorem of forcing for Boolean-valued models \cite{hamkins2012well} states that for any complete Boolean algebra $B$, $V^B$ is a Boolean-valued model of $\mathsf{ZFC}$. Since, in type theory, a type universe \lstinline{Type u} takes the place of the standard universe $V$, the analogous statement in our setting is that for every complete Boolean algebra $\B$, \lstinline{bSet 𝔹} is a Boolean-valued model of $\mathsf{ZFC}$.

Bell \cite{bell2011set} gives an extremely detailed account of the verification of the $\mathsf{ZFC}$ axioms, and we faithfully followed his presentation for this part of the formalization. Most of it is routine. We describe some aspects of \lil{bSet 𝔹} which are revealed by this verification.

\paragraph*{Check-names}
\begin{defn}\label{def-check}
  From the definitions of \lil{pSet} and \lil{bSet}, one immediately sees that there is a canonical map \lil{check : pSet → bSet 𝔹}, defined by
 \begin{lstlisting}
def check : pSet → bSet 𝔹
| ⟨α,A⟩ := ⟨α, λ a, check (A a), λ a, ⊤⟩
\end{lstlisting} We call members of the image of \lil{check} \emph{check-names},\footnote{This terminology is standard, c.f. \cite{hamkins2012well, moore2019method}.} after the usual diacritic notation \lstinline{x̌} for \lil{check (x : pSet)}. These are also known as \emph{canonical names}, as they are the canonical representation of standard two-valued sets inside a Boolean-valued model of set theory.\footnote{We were pleased to discover Lean's support for custom notation allowed us to declare the Unicode modifier character \texttt{U+030C} ($\check{\hspace{1mm}}$) as a postfix operator for \texttt{check}.}
\end{defn}

\paragraph*{The axiom of infinity}
$\omega$ \lil{: bSet 𝔹} is $\check{\omega}$. $\omega$ is defined in \lil{pSet} to be the collection of all finite von Neumann ordinals, which are defined by induction on $\mathbb{N}$. While it is easy to show $\check{\omega}$ satisfies the axiom of infinity
\begin{lstlisting}
def axiom_of_infinity_spec (u : bSet 𝔹) : 𝔹 :=
  (∅∈ᴮ u) ⊓ (⨅i_x, ⨆i_y, (u.func i_x ∈ᴮ u.func i_y))
\end{lstlisting}
it can furthermore be shown to satisfy the universal property of $\omega$, which says that $\omega$ is a subset of any set which contains $\emptyset$ and is closed under the successor operation $x \mapsto x \cup \{x\}$.

\paragraph*{The axiom of powerset}
\begin{defn} \label{def-powerset}
  Fix a $\B$-valued set \lil{x = ⟨α, A, b⟩}. Let \lil{χ : α → 𝔹} be a function. The subset of \lil{x} associated to \lil{χ} is a \lil{𝔹}-valued set $\widetilde{\chi}$ defined as follows:
\begin{lstlisting}
def set_of_indicator {x} (χ : x.type → 𝔹) := ⟨x.type, x.func, χ⟩
\end{lstlisting}

The \textbf{powerset} $\mathcal{P}(x)$ of $x$ is defined to be the following \lil{𝔹}-valued set, whose underlying type is the type of all functions \lil{x.type → 𝔹}:
\begin{lstlisting}
def bv_powerset (u : bSet 𝔹) : bSet 𝔹 :=
⟨u.type → 𝔹, λ f, set_of_indicator f, λ f, set_of_indicator f ⊆ᴮ u⟩
\end{lstlisting}
\end{defn}

\paragraph*{The axiom of choice}
Following Bell, we verified Zorn's lemma, which is provably equivalent over $\mathsf{ZF}$ to the axiom of choice. As is the case with \lil{pSet}, establishing the axiom of choice requires the use of a choice principle from the metatheory. This was the most involved part of our verification of the fundamental theorem of forcing, and relies on the technical tool of \emph{mixtures}, which allow sequences of $\B$-valued sets to be ``averaged'' into new ones, and the \emph{maximum principle}, which allows existentially quantified statements to be instantiated without changing their truth-value.

\paragraph*{The smallness of $\B$}
We end this section by remarking that the ``smallness'' (or more precisely, the fact that $\B$ lives in the same universe of types out of which \lil{bSet 𝔹} is being built) is essential in making \lstinline{bSet 𝔹} a model of $\mathsf{ZFC}$. It is required for extracting the witness needed for the maximum principle, and is also required to even define the powerset operation, because the underlying type of the powerset is the function type of all maps into \lstinline{𝔹}.

\section{Forcing}
\label{sect:forcing}

\subsection{Representing Lean's ordinals inside \lil{pSet} and \lil{bSet}}
The treatment of ordinals in \lil{mathlib} associates a class of ordinals to every type universe, defined as isomorphism classes of well-ordered types, and includes interfaces for both well-founded and transfinite recursion. Lean's ordinals may be represented inside \lil{pSet} by defining a map \lil{ordinal.mk : ordinal → pSet} via transfinite recursion; it is nothing more than the von Neumann definition of ordinals. In pseudocode,
\begin{lstlisting}
def ordinal.mk : ordinal → pSet
| 0 := ∅
| succ ξ := pSet.succ (ordinal.mk ξ) -- (mk ξ ∪ {mk ξ})
| is_limit ξ := ⋃ η < ξ, (ordinal.mk η)
\end{lstlisting}
Composing by \lil{check} (\autoref{def-check}) yields a map \lil{check ∘ ordinal.mk : ordinal → bSet 𝔹}. (We could just as well have defined \lstinline{ordinal.mk' : ordinal → bSet 𝔹} analogously to \lstinline{ordinal.mk} without reference to \lil{check}, such that \lstinline{ordinal.mk' = check ∘ ordinal.mk}; the point is that there is a link between the metatheory's notion of size and order with that of the forcing extension.)

Cardinals in Lean are defined separately from ordinals as bijective equivalence classes of types, but are canonically represented by ordinals which are not bijective with any predecessor. We let \lil{aleph : ordinal → ordinal} index these representatives. For the rest of this section, unadorned alephs (e.g. ``$\aleph_2$'') will mean either an ordinal of the form \lil{aleph ξ} or a choice of representative from the isomorphism class of well-ordered types, and checked alephs (e.g. ``$\check{\aleph_2}$'') will mean the \lil{check ∘ ordinal.mk} of that ordinal.

\subsection{The Cohen poset and the regular open algebra}
Forcing with partial orders and forcing with complete Boolean algebras are related by the fact that every poset of forcing conditions can be embedded into a complete Boolean algebra as a dense suborder. This will be the case for our forcing argument: our Boolean algebra is the algebra of regular opens on $2^{\aleph_2 \times \mathbb{N}}$ (we identify this space with the subsets of $\aleph_2 \times \mathbb{N}$), and the poset of forcing condition embeds in this Boolean algebra as a dense suborder.

\begin{defn}
  The \textbf{Cohen poset} for adding $\aleph_2$-many Cohen reals is the collection of all finite partial functions $\aleph_2 \times \mathbb{N} \to \mathbf{2}$, ordered by reverse inclusion.
\end{defn}

In the formalization, the Cohen poset is represented as a \lstinline{structure} with three fields:
\begin{lstlisting}
structure 𝒞 : Type :=
  (ins : finset (ℵ₂.type × ℕ))
  (out : finset (ℵ₂.type × ℕ))
  (H : ins ∩ out = ∅)
\end{lstlisting}

That is, we identify a finite partial function $f$ with the triple \lil{⟨f.ins, f.out, f.H⟩}, where \lil{f.ins} is the preimage of $\{1\}$, \lil{f.out} is the preimage of $\{0\}$, and \lil{f.H} ensures well-definedness. While $f$ is usually defined as a finite partial function, we found that in practice $f$ is really only needed to give a finite partial specification of a subset of $\aleph_2 \times \mathbb{N}$ (i.e. a finite set \lil{f.ins} which \emph{must} be in the subset, and a finite set \lil{f.out} which \emph{must not} be in the subset), and chose this representation to make that information immediately accessible.

\begin{defn}
  Let $X$ be a topological space, and for any open set $U$, let $U^\perp$ denote the complement of the closure of $U$. The \textbf{regular open algebra} of a topological space $X$, written $\operatorname{RO}(X)$, is the collection of all open sets $U$ such that $U = (U^\perp)^\perp$, equipped with the structure of a complete Boolean algebra, with $x \sqcap y := x \cap y$, $x \sqcup y := ((x \cup y)^\perp)^\perp$, $\neg x := x^\perp$, and $\bigsqcup x_i := ((\bigcup x_i)^\perp)^\perp$.
\end{defn}

The Boolean algebra which we will use for forcing $\neg\mathsf{CH}$ is $\operatorname{RO}(2^{\aleph_2 \times \mathbb{N}})$. Unless stated otherwise, for the rest of this section, we put $\B := \operatorname{RO}(2^{\aleph_2 \times \mathbb{N}})$.

\begin{defn}
  We define the \textbf{canonical embedding} of the Cohen poset into $\B$ as follows:
  \begin{lstlisting}
def ι : 𝒞 → 𝔹 := λ p, {S | p.ins ⊆ S ∧ p.out ⊆ - S}
\end{lstlisting}
\end{defn}
That is, we send each \lil{c : 𝒞} to all the subsets which satisfy the specification given by \lil{c}. This is a clopen set, hence regular. Crucially, this embedding is \emph{dense}:
\begin{lstlisting}
lemma 𝒞_dense {b : 𝔹} (H : ⊥ < b) : ∃ p : 𝒞, ι p ≤ b
\end{lstlisting}
Recalling that $\leq$ in $\B$ is subset-inclusion, we see that this is essentially because the image of $\iota : \mathcal{C} \to \B$ \emph{is} the standard basis for the product topology. Our chosen encoding of the Cohen poset also made it easier to perform this identification when formalizing this proof.
\subsection{Adding $\aleph_2$-many distinct Cohen reals} \label{subsect:cohen-reals}
As we saw in \autoref{def-powerset}, for any $\B$-valued set $x$, characteristic functions into $\B$ from the underlying type of $x$ determine $\B$-valued subsets of $x$. While the ingredients $\aleph_2$ and $\mathbb{N}$ for $\B$ are types and thus external to \lil{bSet 𝔹}, they are represented nonetheless inside \lil{bSet 𝔹} by their check-names $\check{\aleph_2}$ and $\check{\mathbb{N}}$, and in fact \lil{ℵ₂} \emph{is} \lil{ℵ₂̌ .type} and \lil{ℕ} \emph{is} $\check{\mathbb{N}}$\lil{.type}. Given our specific choice of $\B$, this will allow us to construct an $\aleph_2$-indexed family of distinct subsets of $\check{\N}$, which we can then convert into an injective function from \lil{ℵ₂̌ } to \lil{ℕ}, \emph{inside} \lil{bSet 𝔹}.

\begin{defn}
  Let $\nu : \aleph_2$. For any $n : \N$, the collection of all subsets of $\aleph_2 \times \N$ which contain $(\nu, n)$ is a regular open of $2^{\aleph_2 \times \N}$, called the \textbf{principal open} $\mathbf{P}_{(\nu, n)}$ over $(\nu, n)$.
\end{defn}

\begin{defn}
  Let $\nu : \aleph_2$. We associate to $\nu$ the $\B$-valued characteristic function $\chi_{\nu} : \N \to \B$ defined by $\chi_{\nu}(n) := \mathbf{P}_{(\nu, n)}$. In light of our previous observations, we see that each $\chi_{\nu}$ induces a new $\B$-valued subset $\widetilde{\chi_{\nu}} \subseteq \check{\N}$. We call $\widetilde{\chi_{\nu}}$ a \textbf{Cohen real}.
\end{defn}
This gives us an $\aleph_2$-indexed family of Cohen reals. Converting this data into an injective function from $\check{\aleph_2}$ to $\mathbb{N}$ inside \lil{bSet 𝔹} requires some care. One must check that $\nu \mapsto \widetilde{\chi_{\nu}}$ is externally injective, and this is where the characterization of the Cohen poset as a dense subset of $\B$ (and moving back and forth between this representation and the definition as finite partial functions) comes in. Furthermore, one has to develop machinery similar to that for the powerset operation to convert an external injective function \lstinline{x.type → bSet 𝔹} to a $\B$-valued set which \lstinline{bSet 𝔹} thinks is a injective function, while maintaining conditions on the intended codomain. Our custom tactics and automation for reasoning inside $\B$ made this latter task significantly easier than it would have been otherwise. We refer the interested reader to our formalization for details.

\subsection{Preservation of cardinal inequalities} \label{subsect:cardinal-inequalities}
So far, we have shown for $\B = \operatorname{RO}(2^{\aleph_2 \times \mathbb{N}})$ that \lil{bSet 𝔹} thinks $\check{\aleph_2}$ is smaller than $\mathcal{P}(\check{\mathbb{N}})$. Although Lean believes there is a strict inequality of cardinals $\aleph_0 < \aleph_1 < \aleph_2$, in general we can only deduce that their representations inside \lil{bSet 𝔹} are subsets of each other: $\top \leq \check{\aleph_0} \subseteq^\B \check{\aleph_1} \subseteq^\B \check{\aleph_2}$. To finish negating $\mathsf{CH}$, it suffices to show that \lstinline{bSet 𝔹} thinks $\check{\aleph_0}$ is strictly smaller than $\check{\aleph_1}$, and that \lstinline{bSet 𝔹} thinks $\check{\aleph_1}$ is a strictly smaller than $\check{\aleph_2}$. That is, for cardinals $\kappa$, we want that the passage from $\kappa$ to $\check{\kappa}$ to preserve cardinal inequalities.

\begin{defn}
  For our purposes, ``$X$ is strictly smaller than $Y$'' means ``there exists no function \lil{f} such that for every \lil{y ∈ Y}, there exists an \lil{x ∈ X} such that \lil{(x,y) ∈ f}''. Thus, ``\lil{X} is strictly smaller than \lil{Y}'' translates to the Boolean truth-value
\begin{center}\lstinline{-(⨆f, (is_func f) ⊓ ⨅y, y ∈ᴮ Y ⟹ ⨆x, x ∈ᴮ X ⊓ (x, y) ∈ᴮ f)}.\end{center} We abbreviate this with ``$X \prec Y$''.
\end{defn}

The condition on an arbitrary $\B$ which ensures the preservation of cardinal inequalities is the \emph{countable chain condition}.

\begin{defn}
We say that $\B$ has the \textbf{countable chain condition} (CCC) if every antichain $\mathcal{A} : I \to \B$ (i.e. an indexed collection of elements $\mathcal{A} := \{a_i\}$ such that whenever $i \neq j, a_i \sqcap a_j = \bot$) has a countable image.
\end{defn}

We sketch the argument that CCC implies the preservation of cardinal inequalities. The proof is by contraposition. Let $\kappa_1$ and $\kappa_2$ be cardinals such that $\kappa_1 < \kappa_2$, and suppose that $\check{\kappa_1}$ is not strictly smaller than $\check{\kappa_2}$. Then there exists some \lil{f : bSet 𝔹} and some $\Gamma > \bot$ such that \lstinline{Γ ≤ (is_func f) ⊓ ⨅y, y ∈ᴮ κ₁̌  ⟹ ⨆x, x ∈ᴮ κ₂̌  ⊓ (x,y) ∈ᴮ f}. Then one can show:
\begin{lstlisting}
lemma AE_of_check_larger_than_check :
∀ β < κ₂, ∃ η < κ₁, ⊥ < (is_func f) ⊓ (η⠀̌, β ⠀̌ ) ∈ᴮ f
\end{lstlisting}
The name of this lemma emphasizes that what was happened here is that, given this $f$ and the assumption that it satisfes some $\forall$-$\exists$ formula inside \lil{bSet 𝔹}, we are able to extract, by virtue of $\check{\kappa_1}$ and $\check{\kappa_2}$ being check-names, a $\forall$-$\exists$ statement in the \emph{metatheory}. Using Lean's choice principle, we can then convert this $\forall$-$\exists$ statement into a function $g : \kappa_2 \to \kappa_1$, such that for every $\beta$, \lstinline{⊥ < (is_func f) ⊓ (g(β)̌ , β ̌ ) ∈ᴮ f}. Since $\kappa_2 > \kappa_1$, it follows from the infinite pigeonhole principle that there exists some $\eta < \kappa_1$ such that the $g^{-1}(\{\eta\})$ is uncountable. Define $\mathcal{A} : g^{-1}(\{\eta\}) \to \B$ by $\mathcal{A}(\beta) :=$ \lil{(is_func f) ⊓ (g(β)̌ , β ̌ ) ∈ᴮ f}. This is an uncountable antichain because if $\beta_1 \neq \beta_2$, then the well-definedness part of \lil{is_func f} ensures that, since $g(\beta_1) = g(\beta_2)$, the truth-value \lil{β₁̌ } $= f(g(\beta_1)) \neq^\B f(g(\beta_2)) =$ \lil{β₂̌ } is $\bot$.

Thus, conditional on showing that $\B = \operatorname{RO}(2^{\aleph_2 \times \mathbb{N}})$ has the CCC, we now have that cardinal inequalities are preserved in \lstinline{bSet 𝔹}. Combining this with the injection $\check{\aleph_2} \preceq \mathcal{P}(\mathbb{N})$, we obtain:
\begin{lstlisting}
theorem neg_CH : ⊤ = (ℕ ≺ (ℵ₁)̌  ⊓ (ℵ₁)̌  ≺ (ℵ₂)̌  ⊓ (ℵ₂)̌  ≼ 𝒫(ℕ))
\end{lstlisting}

The arguments sketched in \autoref{subsect:cohen-reals} and \autoref{subsect:cardinal-inequalities} form the heart of the forcing argument. Their proofs involve taking objects in \lil{Type u} and \lil{bSet 𝔹}, constructing corresponding objects on the other side, and reasoning about them in ordinary and $\B$-valued logic simultaneously to determine cardinalities in \lstinline{bSet 𝔹}. We have omitted many details from our discussion, but of course, all the proofs have been formally verified.

\subsection{The unprovability of $\mathsf{CH}$}
We conclude this section by briefly describing how the previous results may be converted into a formal proof of the unprovability of $\mathsf{CH}$. We work in a conservative expansion $\mathsf{ZFC}'$ of $\mathsf{ZFC}$ with an expanded language $L_{\mathsf{ZFC}'}$ with symbols for pairing, union, powerset, and $\omega$. We define $\mathsf{ZFC}'$ to be precisely the $\mathsf{ZFC}$ axioms which were verified in the fundamental theorem of forcing, along with specifications for the new function symbols. $\mathsf{CH}$ can then be written as a deeply-embedded $L_{\mathsf{ZFC}'}$ sentence (note the use of de Bruijn indices for variables)
\begin{lstlisting}
def CH : sentence L_ZFC' := ¬ ∃' ∃' (ω ≺ &1) ⊓ (&1 ≺ &0) ⊓ (&0 ≼ 𝒫(ω))
\end{lstlisting}
where \lil{≺} and \lil{≼} are abbreviations with the same meaning as in the previous section. Then proving \lstinline{bSet 𝔹 ⊨ ZFC' + ¬CH } is a straightforward matter of checking that sentences are interpreted correctly as Boolean truth values which we have already proved to be $\top$. Applying the contrapositive of the Boolean-valued soundness theorem yields the result.

\section{Transfinite combinatorics and the countable chain condition}
\label{sect:ccc}
What remains now is to prove that $\operatorname{RO}(2^{\aleph_2 \times \mathbb{N}})$ has the CCC. There are several ways forward; we chose a very general proof using the
$\Delta$-system lemma to show more generally that the product of topological spaces satisfies the CCC if every finite subproduct does. Our proof follows Kunen \cite{kunen2014set}.

\subsection{The $\Delta$-system lemma}

\begin{defn}
  A family $(A_i)_i$ of sets is called a \textbf{$\Delta$-system} (or a \textbf{sunflower} or \textbf{quasi-disjoint}) if there is a set $r$, called the \textbf{root} such that whenever $i \ne j$ we have $A_i \cap A_j = r$.
\end{defn}
\begin{lstlisting}
def is_delta_system {α ι : Type*} (A : ι → set α) :=
∃(root : set α), ∀{{x y}}, x ≠ y → A x ∩ A y = root
\end{lstlisting}

The $\Delta$-system lemma states that if we have an uncountable family of finite sets, there is an uncounbtable subfamily which forms a $\Delta$-system.
In Lean this is formulated as follows. (\lstinline{restrict A t} is the restriction of the collection \lstinline{A} to \lstinline{t}).
\begin{lstlisting}
theorem delta_system_lemma_uncountable {α ι : Type*}
  (A : ι → set α) (h : cardinal.omega < mk ι)
  (h2A : ∀i, finite (A i)) : ∃(t : set ι),
  cardinal.omega < mk t ∧ is_delta_system (restrict A t)
\end{lstlisting}
This theorem follows from the following more general statement, taking $\kappa=\aleph_0$ and $\theta=\aleph_1$ (for cardinal numbers the operation \lstinline{c ^< κ} or $c^{<\kappa}$ is the supremum of $c^\rho$ for $\rho<\kappa$).
\begin{lstlisting}
theorem delta_system_lemma {α ι : Type u} {κ θ : cardinal}
  (hκ : cardinal.omega ≤ κ) (hκθ : κ < θ) (hθ : is_regular θ)
  (hθ_le : ∀(c < θ), c ^< κ < θ) (A : ι → set α)
  (hA : θ ≤ mk ι) (h2A : ∀i, mk (A i) < κ) :
  ∃(t : set ι), mk t = θ ∧ is_delta_system (restrict A t)
\end{lstlisting}

We omit the proof, referring the interested reader to \cite{kunen2014set} or the formalization.

\subsection{$\operatorname{RO}(2^{\aleph_2 \times \mathbb{N}})$ has the countable chain condition}

\begin{defn}
  We say that a topological space $X$ satisfies the countable chain condition if every family of pairwise disjoint open sets is countable.
\end{defn}

We first give a sufficient condition for a product of topological spaces to satisfy the countable chain condition.
\begin{thm}
If we have a family $(X_i)_{i\in I}$ of topological spaces, then $\prod_{i\in I} X_i$ has the countable chain condition if for every finite $J\subseteq I$ the product $\prod_{i\in J} X_i$ has the countable chain condition.
\end{thm}
\begin{proof}
For the proof, suppose we had an uncountable family of pairwise disjoint open subsets $U_k$ of $\prod_{i\in I} X_i$. By shrinking $U_k$, we may assume that each $U_k$ is a basic open set of the form $\prod_{i\in F_k} U_{k,i} \times \prod_{i \not\in F_k} X_i$ for some finite set $F_k$. Now the $(F_k)_k$ form a uncountable family of finite sets, so by the $\Delta$-system lemma we know that there is an uncountable family $K$ of indices such that $(F_k)_{k\in K}$ forms a $\Delta$-system with root $J$. Now we can take the projections $\pi(U_k)$ onto $\prod_{i\in J}X_i$ for $k\in K$. We can show this forms an uncountable disjoint family of opens in $\prod_{i\in J}X_i$, contradicting the assumption.
\end{proof}

With this, the rest of the proof that $\B=\operatorname{RO}(2^{\aleph_2 \times \mathbb{N}})$ has the CCC is easy: since every finite product $2^J$ is a finite topological space, and so satisfies the CCC, it follows that the space $2^{\aleph_2 \times \mathbb{N}}$ satisfies the CCC. Also, if a topological space $X$ satisfies the CCC then the algebra of regular opens satisfies the CCC, since every antichain of regular opens forms a family of disjoint open sets. Thus, we have shown:
\begin{lstlisting}
theorem 𝔹_CCC : CCC (regular_opens (set(ℵ₂.type × ℕ)))
\end{lstlisting}

\section{Related work}
\paragraph*{First-order logic, soundness, and completeness} There are many existing formalizations of first-order logic. Shankar \cite{shankar1997metamathematics} used a deep embedding of first-order logic to formalize incompleteness theorems. Harrison gives a deeply-embedded implementation of first-order logic in HOL Light \cite{harrison1998formalizing} and a proof-search style account of the completeness theorem in \cite{harrison2009handbook}. Margetson \cite{Ridge2005AMV} and Schlichtkrull \cite{schlichtkrull2018formalization} use the same argument for the completeness theorem in Isabelle/HOL, while Berghofer \cite{FOL-Fitting-AFP} (in Isabelle) and Ilik \cite {ilik2010constructive} (in Coq) use canonical term models.

\paragraph*{Set theory and forcing}
Set theory is a common target for formalization. Notably, a large body of formalized set theory has been completed in Isabelle/ZF, led by Paulson and his collaborators \cite{paulson1996mechanizing, paulson1993set, paulson2002reflection}. Most relevantly, this includes a formalization of the relative consistency of the axiom of choice with $\mathsf{ZF}$ \cite{paulson2003relative}. Building on this, Gunther, Pagano, and Terraf have begun formalizing the basic ingredients of forcing \cite{gunther2018first, gunther2019mechanization}, taking the more conventional approach of generic extensions of countable transitive models.

Our tactic library for Boolean-valued logic was inspired by work of Hudon \cite{Hudon2015TheUM} on Unit-B, using similar techniques to embed a proof language for temporal logic \cite{unitb}. It was pointed out to the authors that a trick similar to \autoref{poset-yoneda} had also been successfully applied in the Metamath library \cite{mario2}.

The work we have described in this paper relies heavily on Lean's \lstinline{mathlib}. In particular, the extensive \lstinline{set_theory} and \lstinline{ordinal} libraries contained nearly everything we needed (including a treatment of cofinalities for the $\Delta$-system lemma), with missing parts easily accessible through existing lemmas. These libraries were originally developed by Carneiro \cite{mario1}, in part to show that Lean proves the existence of infinitely many inaccessible cardinals.

\section{Conclusions and future work}
\paragraph*{Reflections on the proof}
As our formalization has shown, for the purposes of a consistency proof, one can perform forcing entirely outside of the set-theoretic foundations in which forcing is usually presented. There is no need to work inside an ambient model of set theory, or to even have a ground model of set theory over which one constructs a forcing extension. Instead, the recursive \emph{name} construction applied to a universe of types is key. The type universe, with its classical two-valued logic and its own notion of ordinals, takes the place of the standard universe of sets. These external ordinals are then represented in the internal ordinals of the forcing extension by indexing the construction of von Neumann ordinals. With a clever choice of forcing conditions $\B$, this representation of ordinals will preserve cardinal inequalities and force an uncountable set beneath $\mathcal{P}(\N)$.

In particular, \lstinline{pSet}, being only another special case of the construction which produces \lstinline{bSet 𝔹}, is no longer a prerequisite for working with \lstinline{bSet 𝔹}, but merely a convenient tool for organizing the check-names---this is the only role it played in the proof. The check-names themselves were actually not necessary either: as we remarked, the canonical map \lil{ordinal → bSet 𝔹} can be defined without reference to them. However, since in all of our sources, \lstinline{pSet} additionally played the role of the universe of types, and an interface for it was readily available in \lstinline{mathlib}, we started our formalization by following the usual arguments, implementing these simplifications as we became aware of them.
\paragraph*{Lessons learned}
\begin{itemize}
\item Originally, we thought set-theoretic arguments involving transfinite/ordinal induction, which are ubiquitous, would be difficult to implement. In practice, Lean's tools for well-founded recursion and the comprehensive treatment of ordinals in \lstinline{mathlib} made the implementation of such arguments painless.
\item Definitions and lemmas should be stated as generally as possible. This maximizes reusability, minimizes redundancy, and by exposing only the information required to complete the proof, improves the performance of automation.
\item One should invest early in domain-specific automation. The formalization of the fundamental theorem was completed using only the first two strategies outlined in \autoref{subsect:proof-language}; the calculations, while tedious, were recorded in our sources and it seemed easier to follow them. If we had followed through on the observations around \autoref{poset-yoneda} and developed the custom tactic library earlier, we would have saved a significant amount of time.
\end{itemize}
\paragraph*{Towards a formal proof of the independence of the continuum hypothesis}

The work we have described in this paper was undertaken as part of the Flypitch project, which aims to produce a formal proof of the independence of the continuum hypothesis. As such, the obvious next goal is a formalization of the consistency of $\mathsf{CH}$. Although it would be possible to do this using Boolean-valued models, we intend to develop the infrastructure necessary to support a proof by forcing with generic extensions, as well as G\"odel's original proof by way of analyzing the constructible universe $\mathsf{L}$.

Although our work includes a formal proof of the unprovability of a version of $\mathsf{CH}$ from a version of the $\mathsf{ZFC}$ axioms in a conservative extension of the language of $\mathsf{ZFC}$, verifying this is easy. What is more interesting is formalizing the equivalence of various common formulations of $\mathsf{ZFC}$ and $\mathsf{CH}$, so that a skeptical user may verify that their preferred version of $\mathsf{CH}$ is unprovable from their preferred version of $\mathsf{ZFC}$. This would require formalizations of the conservativity of commonly-used extensions of $\mathsf{ZFC}$, and of the equivalence of the various ways to say that one set is strictly smaller than another. 
The proof of the completeness theorem already required formalizing nontrivial conservativity statements, which shows that our framework is well-equipped to support such results.

Although the stated goal of our project is to achieve a formal proof of the independence of the continuum hypothesis, we also intend to develop reusable libraries for set theory and mathematical logic. We have completed a formalization of forcing, but are nowhere near completing a library which a set theorist could use to verify their research. Just as, more than 50 years ago, Cohen's proof marked the beginning of modern research in set theory, a formal proof of the independence of the continuum hypothesis will only mark the beginning of an integration of formal methods into modern research in set theory. This will require robust interfaces for handling the diverse range of forcing arguments and for reasoning about the consistency strengths of various extensions of $\mathsf{ZFC}$, so that---to paraphrase Kanamori \cite{kanamori1996mathematical, kanamori2008higher}---deeply-embedded notions of truth and relative consistency become matters of routine manipulation as in algebra. Our work demonstrates that such tasks are well within the scope of modern interactive theorem provers.

\section{References}

\bibliography{flypitch-itp-2019}

\end{document}